\documentclass[11pt]{article}

\usepackage{graphicx,xspace,amsmath,amssymb,fullpage,amsthm}
\usepackage[usenames,dvipsnames]{color}
\usepackage{algorithmicx,algorithm}
\usepackage[noend]{algpseudocode}

\newcommand{\VC}{\textsf{Vertex Cover}\xspace}
\newcommand{\CVD}{\textsf{Cluster Vertex Deletion}\xspace}
\newcommand{\CN}{\textsf{Chromatic Number}\xspace}

\newcommand{\FVS}{\textsf{Feedback Vertex Set}\xspace}
\newcommand{\ignore}[1]{}

\newtheorem{theorem}{Theorem}
\newtheorem{lemma}{Lemma}
\newtheorem{task}{Task}
\newtheorem{definition}{Definition}

\begin{document}

\title{Fast Dynamic Graph Algorithms for Parameterized Problems\thanks{
A preliminary version of this paper appears in the proceedings of SWAT 2014.
}}

\author{Yoichi Iwata
	\thanks{
		Department of Computer Science,
		Graduate School of Information Science and Technology,
		The University of Tokyo.
		\texttt{y.iwata@is.s.u-tokyo.ac.jp}
	}
\and Keigo Oka
	\thanks{
		Department of Computer Science,
		Graduate School of Information Science and Technology,
		The University of Tokyo.
		\texttt{ogiekako@is.s.u-tokyo.ac.jp}
	}
}

\date{}

\maketitle

\begin{abstract}
Fully dynamic graph is a data structure that (1) supports edge insertions and deletions and (2) answers problem specific
queries.
The time complexity of (1) and (2) are referred to as the update time and the query time
respectively.
There are many researches on dynamic graphs whose update time and query time are
$o(|G|)$, that is, sublinear in the graph size.
However, almost all such researches are for problems in P.
In this paper, we investigate dynamic graphs for NP-hard problems exploiting the
notion of fixed parameter tractability (FPT).

We give dynamic graphs for \VC and \CVD parameterized by the solution size $k$.
These dynamic graphs achieve almost the best possible
update time $O(\mathrm{poly}(k)\log n)$ and the query time
$O(f(\mathrm{poly}(k),k))$, where $f(n,k)$ is the time complexity of any static
graph algorithm for the problems.
We obtain these results by dynamically maintaining an approximate solution which can be used to construct a
small problem kernel.
Exploiting the dynamic graph for \CVD,
as a corollary, we obtain a quasilinear-time (polynomial) kernelization
algorithm for \CVD.
Until now, only quadratic time kernelization algorithms are known
for this problem.

We also give a dynamic graph for \CN parameterized by the solution size of
\CVD, and a dynamic graph for bounded-degree \FVS parameterized by the solution
size. Assuming the parameter is a constant, each dynamic graph can be updated in
$O(\log n)$ time and can compute a solution in $O(1)$ time.
These results are obtained by another approach.
\end{abstract}

\section{Introduction}
\subsection{Background}
\subsubsection{Parameterized Algorithms}
Assuming $\textrm{P}\neq \textrm{NP}$,
there are no polynomial-time algorithms solving NP-hard problems. On the other
hand, some problems are efficiently solvable when a certain parameter, e.g.
the size of a solution, is small.
\emph{Fixed parameter tractability} is one of the ways to capture such a
phenomenon.

A problem is in the class \emph{fixed parameter
tractable (FPT)} with respect to a parameter $k$ if there is an algorithm that
solves any problem instance of size $n$ with parameter $k$ in $O(n^df(k))$ time (\emph{FPT
time}), where $d$ is a constant and $f$ is some computable function.

\subsubsection{Dynamic Graphs}
(Fully) dynamic graph is a data structure that supports edge insertions, edge deletions, and answers certain problem
specific queries.
There are a lot of theoretical research on dynamic graphs for problems that
belong to $\mathrm{P}$, such as
\textsf{Connectivity}~\cite{DBLP:conf/stoc/HenzingerK95,DBLP:journals/jacm/HolmLT01,DBLP:conf/stoc/Thorup00,DBLP:conf/soda/Wulff-Nilsen13a,DBLP:conf/focs/EppsteinGIN92,DBLP:conf/stoc/PatrascuD04},
\textsf{$k$-Connectivity}~\cite{DBLP:journals/jacm/HolmLT01,DBLP:conf/focs/EppsteinGIN92},
\textsf{Minimum Spanning Forest}~\cite{DBLP:journals/jacm/HolmLT01,DBLP:conf/focs/EppsteinGIN92}, 
\textsf{Bipartiteness}~\cite{DBLP:journals/jacm/HolmLT01,DBLP:conf/focs/EppsteinGIN92}, 
\textsf{Planarity
Testing}~\cite{DBLP:conf/esa/ItalianoPR93,DBLP:conf/focs/EppsteinGIN92,DBLP:conf/stoc/Poutre94},
\textsf{All-pairs Shortest
Path}~\cite{DBLP:conf/swat/Thorup04,DBLP:conf/stoc/DemetrescuI03,DBLP:conf/stoc/Thorup05,DBLP:conf/esa/RodittyZ04,DBLP:conf/focs/RodittyZ04}
and \textsf{Directed Connectivity}~\cite{DBLP:conf/focs/DemetrescuI00,DBLP:conf/soda/Roditty03,DBLP:conf/focs/RodittyZ02,DBLP:conf/stoc/RodittyZ04,DBLP:conf/focs/Sankowski04},
and races for faster algorithms are going on.

On the contrary there have been few research on dynamic graphs related to FPT
algorithms. 
To the best of our knowledge, 
a dynamic data structure for counting subgraphs in sparse
graphs proposed by Zden\v ek Dv\v o\'rak and Vojt\v ech T\r
uma~\cite{DBLP:conf/wads/DvorakT13} and 
a dynamic data structure for tree-depth decomposition proposed by Zden\v ek
Dvo\v r\'ak, Martin Kupec and Vojt\v ech T\r
uma~\cite{DBLP:journals/corr/DvorakKT13} are only such dynamic graphs.
Both data structures support insertions and deletions of
edges, and compute the solution of the problems in time depending only
on $k$, where $k$ is the parameter of the problem.
For a fixed property expressed in monadic second-order logic, the dynamic
graph in~\cite{DBLP:journals/corr/DvorakKT13} also can answer whether the current
graph has the property.
For both algorithms, hidden constants are (huge) exponential in $k$.
In particular, update time of both algorithms become super-linear in graph
size $n$ even if $k$ is very small, say $O(\log \log n)$.

\subsection{Our Contribution}
In this paper, we investigate dynamic data structures for basic graph problems
in FPT. Table~\ref{tab:complexity} shows the problems we deal with and the time
complexities of the algorithms.

\subsubsection{Dynamic Graph for Vertex Cover and Cluster Vertex Deletion}
In Section~\ref{sec:VC} and~\ref{sec:CVD}, we present fully dynamic graphs for
\VC and \CVD, respectively.
Both dynamic data structures support
additions and deletions of edges, and can answer the solution of the
problem in time depending only on the solution size $k$.

\begin{table}[tb]
     \center
	  \begin{tabular}{| c || c | c | c |}
	    \hline
	  	Problem	& Update Time	& Query Time	& Section				\\
	  	\hline														
	  	\VC		& $O(k^2)$		& $f_{VC}(k^2,k)$	& \ref{sec:VC} 			\\
	  	\CVD	& $O(k^8+k^2\log n)$			& $f_{CVD}(k^5,k)$			& \ref{sec:CVD}
	  	\\
	  	\CVD	& $O(8^kk^6)$			& $O(1)$			& \ref{sec:CVD2}		\\
	  	\CN 	& $O(2^{2^k}\log n)$\footnotemark & $O(1)$ &
	  	\ref{sec:CN}
	  	\\
	  	\FVS	& $O(7.66^kk^3 + 2^kk^3d^3\log n)$			& $O(1)$			&
	  	\ref{sec:feedback}	\\
	  	\hline
	  \end{tabular}
  \caption{The time complexities of the dynamic graphs in this paper.
  $d$ is the degree bound, and $n$ is the number of the vertices.
  The parameter for \CN is cvd number (the size of a minimum cluster vertex
  deletion), and parameters for the other problems are its solution size.}
  \label{tab:complexity}
\end{table}
\footnotetext{More precisely, it is
$O(B_{2k}(4^kk^3 + \log n))$ as proved in
Section~\ref{sec:CN}. ($B_n$ is the \emph{Bell number} for $n$, the number of
ways to divide a set with $n$ elements.)}

For the dynamic graph for \VC, the time complexity of an edge addition or
deletion is $O(k^2)$ and the one of a query is $f_{VC}(k^2,k)$ where $f_{VC}(n,
k)$ is the time complexity of any static algorithm for \VC on a graph of size $n$.

For the dynamic graph for \CVD, the time complexity of an update is $O(k^8\log
n)$ and the one of a query is $f_{CVD}(k^5,k)$ where
$f_{CVD}(n, k)$ is the time complexity of any static algorithm for \CVD on a
graph of size $n$.
The extra $\log n$ factor arises because we use persistent
data structures to represent some vertex sets. This enables us to copy a set in
constant time.

Note that the time complexity of an update is
$\mathrm{poly}(k)\mathrm{polylog}(n)$ for both algorithms, instead of an
exponential function in $k$.
As for the time complexity of a query, its exponential term in $k$ is no more
than any static algorithms.

Let us briefly explain how the algorithms work.
Throughout the algorithm, we keep an approximate solution.
When the graph is updated, we efficiently construct a $\mathrm{poly}(k)$ size kernel by exploiting the approximate
solution, and then compute a new approximate solution on this kernel.
Here, we compute not an exact solution but an approximate
solution to achieve the update time polynomial in $k$.
To answer a query, we apply a static exact algorithm to the kernel.

To see goodness of these algorithms, consider
the situation such that a query is applied for every $r$ updates.
A trivial algorithm answers a query by running a static
algorithm.
Let the time complexity of the static algorithm be
$O(f(n,k))$.
In this situation, to deal with consecutive $r$ updates and one query, our
algorithm takes $O(r\mathrm{poly}(k)\mathrm{polylog}(n) +
f(\mathrm{poly}(k),k))$ time, and the trivial algorithm takes $O(f(n,k))$ time.
For example, let $f(n,k) = c^k + kn$ be the time complexity of the static
algorithm. (The time complexity of the current best FPT algorithm for \VC is
$O(1.2738^k + kn)$~\cite{DBLP:journals/tcs/ChenKX10}.)
Then if $r=\sqrt{n}$ and $c^k=\sqrt{n}$, the time complexity for the dynamic
graph algorithm is $\sqrt{n}\mathrm{polylog}(n)=o(n)$, sublinear in $n$.
That is, our algorithm works well even if the number of queries is fewer than
the number of updates.
This is an advantage of the polynomial-time update.
If $r=1$, our algorithm is faster than the trivial algorithm whenever $c^k < n$.
Even if $c^k$ is the dominant term, our algorithm is never
slower than the trivial algorithm.

Let us consider the relation between our results and the
result by Dvo\v r\'ak, Kupec and T\r
uma~\cite{DBLP:journals/corr/DvorakKT13}.
The size of a solution of \VC is called \emph{vertex cover number}, and 
the size of a solution of \CVD is called \emph{cluster vertex deletion number}
(\emph{cvd number}).
It is easy to show that tree-depth can be arbitrarily large
even if cvd number is fixed and vice versa.
Thus our result for \CVD is not included in their result.
On the other hand, tree-depth is
bounded by $\mathrm{vertex~cover~number} + 1$.
Thus their result indicates that \VC can be dynamically computed in $O(1)$ time if
vertex cover number is a constant. However, if it is not a constant, say
$O(\log \log n)$, the time complexity of their algorithm becomes no longer
sublinear in $n$.
The time complexity of our algorithm for \VC is further moderate as noted
above.

As an application of the dynamic graph for \CVD, we can obtain a
quasilinear-time kernelization algorithm for \CVD.
To compute a problem kernel of a graph $G=(V,E)$, starting from an empty graph, we
iteratively add the edges one by one while updating an approximate solution.
Finally, we compute a kernel from the approximate solution.
As shown in
Section~\ref{sec:CVD}, the size of the problem kernel is $O(k^5)$ and the time for
an update is $O(k^8\log |V|)$.
Thus, we obtain a polynomial kernel in $O(k^8|E|\log |V|)$ time.

Protti, Silva and
Szwarcfiter~\cite{DBLP:journals/mst/ProttiSS09}
proposed a linear-time kernelization algorithm for \textsf{Cluster Editing}
applying modular decomposition techniques.
On the other hand, to the best of our knowledge, for \CVD, only quadratic time
kernelization algorithms~\cite{DBLP:journals/mst/HuffnerKMN10} are known (until now). Though \CVD and \textsf{Cluster
Editing} are similar problems, it seems that their techniques cannot be directly applied to obtain
a linear-time kernelization algorithm for \CVD.

\subsubsection{Dynamic Graph for Chromatic Number Parameterized by CVD Number}
The study of problems parameterized by cvd number was initiated by
Martin Doucha and Jan Kratochv\'il~\cite{DBLP:conf/mfcs/DouchaK12}.
They studied the fixed parameter tractability of basic graph problems related to
coloring and Hamiltonicity, and proved that the problems \textsf{Equitable
Coloring}, \textsf{Chromatic Number}, \textsf{Hamiltonian Path} and 
\textsf{Hamiltonian Cycle} are in FPT parameterized by cvd number.

In this paper, we also obtained a fully dynamic data structure for \CN parameterized by cvd
number. Assuming the cvd number is a
constant, the time complexity of an update and a query is $O(1)$.
In our algorithm, we maintain not only a minimum cluster
vertex deletion but also more detailed
information, equivalent vertex classes in each cluster.
We consider two vertices in a same cluster are equivalent if their neighbors in the current solution are exactly same.
To update such underlying data structures efficiently, we present another dynamic graph for \CVD in
Section~\ref{sec:CVD2}.
Unlike the algorithms in section~\ref{sec:VC} and ~\ref{sec:CVD}, this algorithm deals with an update in exponential time in $k$, and it seems difficult to make it polynomial by the need to maintain equivalent classes.

Then, we design a dynamic graph for \CN in Section~\ref{sec:CN}.
In the algorithm, for each update we consider each possible coloring for the current minimum cluster vertex deletion and compute the minimum number of colors to color the other vertices exploiting the equivalent classes and a flow algorithm. 
Our algorithm is based on a static algorithm in~\cite{DBLP:journals/jcss/GuoGHNW06}.

\subsubsection{Dynamic Graph for Bounded-Degree Feedback Vertex Set}
Finally we present a fully dynamic data
structure for bounded-degree \FVS in Section~\ref{sec:feedback}.
Despite the restriction of the degree, 
we believe the result is still worth mentioning because the algorithm is never
obvious. The algorithm is obtained by exploiting a theorem
in~\cite{DBLP:journals/jcss/GuoGHNW06} and a classic Link-Cut Tree data structure
introduced by Sleator and Tarjan~\cite{DBLP:journals/jcss/SleatorT83}.
As with~\cite{DBLP:journals/jcss/GuoGHNW06}, we use the idea of 
\emph{iterative compression}.
Iterative compression is the technique introduced by Reed, Smith and Vetta~\cite{DBLP:journals/orl/ReedSV04}. Its central idea is to iteratively compute a minimum solution with size $k$ making use of a solution with size $k+1$.

This algorithm also takes exponential time in $k$ for an update mainly because we consider all $O(2^k)$ possibility of $X\cap X'$ where $X$ is the current solution and $X'$ is the updated solution.

It seems an interesting open question whether it is possible to construct an efficient
dynamic graph without the degree restriction.

\section{Notations} \label{sec:preliminaries}
Let $G=(V,E)$ be a simple undirected graph with vertices $V$ and edges $E$.
We consider that each edge in $E$ is a set of vertices of size two.
Let $|G|$ denote the \emph{size} of the graph $|V| + |E|$.
The \emph{neighborhood} $N_G(v)$ of a vertex $v$ is $\left\{ u\in V\mid
\{u,v\}\in E \right\}$, and the neighborhood $N_G(S)$ of a vertex set
$S\subseteq V$ is $ \bigcup_{v\in S}N_G(v)\setminus S$.
The \emph{closed neighborhood} $N_G[v]$ of a vertex $v$ is $N_G(v) \cup \{v\}$,
and the closed neighborhood $N_G[S]$ of a vertex set $S\subseteq V$ is $N_G(S)
\cup S$.
Let the \emph{incident edges} $\delta_G(v)$ of a vertex $v$ be the set of edges
incident to the vertex $v$.
The \emph{cut edges} $\delta_G(S,T)$ between two disjoint vertex subsets $S$ and $T$ are $\{\{u,v\}\in E\mid u\in S,
v\in T\}$.
We denote the degree of a vertex $v$ by $d_G(v)$.
We omit the subscript if the graph is apparent from the context.
The \emph{induced subgraph} $G[S]$ of a vertex set $S$ is the graph $(S, \left\{
e \in E\mid e\subseteq S \right\})$.
For an edge subset $F\subseteq E$, let $G-F$ be the graph $(V,E\setminus
F)$. 

By default, we use $k(G)$ or $k$ to denote the parameter value of the current
graph $G$.
When an algorithm updates a graph $G$ to $G'$, we use $k = \max\{k(G),k(G')\}$
as a parameter. Note that, 
$k(G)$ and $k(G')$ are not greatly different in most problems. In particular, it
is easy to prove that for all problems we deal with in this paper, $k(G')$ is at most
$k(G) + 1$.

\section{Dynamic Graph for Vertex Cover}\label{sec:VC}
Let $G=(V,E)$ be a graph.
{\sf Vertex Cover} is the problem of finding a minimum set of vertices that covers all edges.
Let $k=k(G)$ be the size of a minimum vertex cover of $G$.
The current known FPT algorithm solving \VC whose exponential function in $k$
is smallest is by Chen, Kanj and Xia~\cite{DBLP:journals/tcs/ChenKX10}, and
its running time is $O(|G|k + 1.2738^k)$.
Let us now state the main result of this section.
\begin{theorem}\label{thm:VC}
There is a data structure representing a graph $G$ which supports the following three operations.
\begin{enumerate}
\item Answers the solution for {\sf Vertex Cover} of the current graph $G$.
\item Add an edge to $G$.
\item Remove an edge from $G$.
\end{enumerate}
Let $k$ be the size of a minimum vertex cover of $G$.
Then the time complexity of an edge addition or removal is $O(k^2)$, and
 of a query is $O(f(k^2,k))$, where $f(|G|,k)$ is the time complexity of any
 static algorithm for \VC on a graph of size $|G|$.
\end{theorem}
Note that the update time is polynomial in $k$, and the exponential term in $k$ of the query
time is same to the one of the static algorithm.

Our dynamic data structure is simply represented as a pair of the graph
$G=(V,E)$ itself and a 2-approximate solution $X\subseteq V$ for \VC of $G$,
that is, we maintain a vertex set $X$ such that $X$ is a vertex cover of $G$ and
$|X|\leq 2k(G)$.

For both query and update, we compute a problem kernel.
To do this, we exploit the fact that we already know rather small vertex cover
$X$.
When an edge $\{u,v\}$ is added to $G$, we add $u$ to $X$ making $X$ a vertex
cover and use Algorithm~\ref{alg:VC} to
compute a new 2-approximate solution $X'$ of $G$.
When an edge is removed from $G$, we also use Algorithm~\ref{alg:VC} to compute
a new 2-approximate solution.
\begin{algorithm}[tb]
\caption{compute a 2-approximate solution}
\label{alg:VC}
\begin{algorithmic}[1]
\State $X_0 := \emptyset$				
\State $V' := \emptyset$				
\ForAll{$x$ in $X$}						
	\If{$d(x) > |X|$}					
		{$X_0 := X_0\cup \{x\}$}	
    \Else								
        {\enspace $V' := V'\cup N[x]$}	
    \EndIf								
\EndFor									
\State $V' := V'\setminus X_0$			
\State $Y := $ 2-approximate solution for {\sf Vertex Cover} of $G[V']$
\label{lne:vc:compute}
\State $X' := X_0\cup Y$				
\end{algorithmic}
\end{algorithm}

\begin{lemma}\label{lem:VC}
Algorithm~\ref{alg:VC} computes a 2-approximate solution $X'$ in $O(k^2)$ time,
where $k=k(G)$.
\end{lemma}
\begin{proof}
Let $X^*$ be a minimum vertex cover of the updated graph.
We have $|X^*| \leq |X|$.
If $x\notin X^*$ for some vertex $x\in X_0$, $N(x)$ must be contained in $X^*$.
Thus it holds that $|X^*|\geq |N(x)|=d(x) > |X|$, which is a contradiction.
Therefore, it holds that $X_0\subseteq X^*$.
At line \ref{lne:vc:compute} of Algorithm \ref{alg:VC}, $V'$ equals to $N[X\setminus
X_0]\setminus X_0$.
Thus we have:
\begin{itemize}
\item[(1)] $X^*\setminus X_0$ is a vertex cover of $G[V']$ because $X_0\cap
V'=\emptyset$ and $X^*$ is a vertex cover of $G$, and
\item[(2)] any vertex cover of $G[V']$ together with $X_0$ covers all edges in
$G$ because all edges not in $G[V']$ are covered by $X_0$.
\end{itemize}
Putting (1) and (2) together, we can prove that $X^*\setminus X_0$ is a
minimum vertex cover of $G[V']$.

Since $Y$ is a 2-approximate solution on $G[V']$ and $X^*\setminus X_0$ is a minimum vertex cover of $G[V']$,
we have $|Y|\leq 2|X^*\setminus X_0|$.
From (2), $X'=X_0\cup Y$ is a vertex cover of $G$.
Thus $X'$ is a 2-approximate solution because $|X'|=|X_0|+|Y|\leq
|X_0| + 2|X^*\setminus X_0| \leq 2|X^*|$.

The size of $X$ is at most $2k+1$, and thus the size of $V'$ at line~\ref{lne:vc:compute}
is $O(k^2)$. Moreover, the number of edges in $G[V']$ is $O(k^2)$ because for
each edge in $G[V']$, at least one endpoint lies on
$X\setminus X_0$ and the degree of any vertex $x$ in $X\setminus X_0$ is at most
$|X|$.
A 2-approximate solution can be computed
in linear time using a simple greedy algorithm~\cite{gary1979computers}.
Thus the total time complexity is $O(k^2)$.
\end{proof}

To answer a query, we use almost the same algorithm as Algorithm~\ref{alg:VC}, but
compute an exact solution at line~\ref{lne:vc:compute} instead of an approximate
solution.
The validity of the algorithm can be proved by almost the same argument.
The bottleneck part is to compute an exact vertex cover of the graph $G[V']$.
Since the size of the solution is at most $k$, we can obtain
the solution in $O(f(k^2,k))$ time where $f(|G|,k(G))$ is the time complexity of
any static algorithm for \VC on a graph of size $|G|$. For example, using the algorithm
in~\cite{DBLP:journals/tcs/ChenKX10}, we can compute the solution in $O(k^3 +
1.2738^k)$ time. We have finished the proof of Theorem~\ref{thm:VC}.

\section{Dynamic Graph for Cluster Vertex Deletion}\label{sec:CVD}
\subsection{Problem Definition and Time Complexity}\label{sec:CVD_prob_time}
A graph is called \emph{cluster graph} if every its connected component is a
clique, or equivalently, it contains no \emph{induced} path with three vertices
($P_3$).
Each maximal clique in a cluster graph is called a \emph{cluster}.
Given a graph, a subset of its vertices is called a \emph{cluster vertex
deletion} if its removal makes the graph a cluster graph. \CVD is the problem to
find a minimum cluster vertex deletion. We call the size of a minimum cluster
vertex deletion as a \emph{cluster vertex deletion number} or a \emph{cvd number} in
short.

There is a trivial algorithm to find a 3-approximate solution for \CVD with time
complexity $O(|E||V|)$~\cite{DBLP:journals/mst/HuffnerKMN10}.
The algorithm greedily finds $P_3$ and adds all the
vertices on the path to the solution until we obtain a cluster graph.
According to~\cite{DBLP:journals/mst/HuffnerKMN10}, it is still open
whether it is possible to improve the trivial algorithm or not.

Let us now state the main result of this section.

\begin{theorem}\label{thm:CVD}
There is a data structure representing a graph $G$ which supports the following three operations.
\begin{enumerate}
\item Answers the solution for {\sf Cluster Vertex Deletion} of the current graph $G$.
\item Add an edge to $G$.
\item Remove an edge from $G$.
\end{enumerate}
Let $k$ be the cvd number of $G$.
Then the time complexity of an edge addition or removal is $O(k^8+k^2\log |V|)$, and
 of a query is $O(f(k^5,k))$, where $f(|G|,k)$ is the time complexity of any
 static algorithm for \CVD on a graph of size $|G|$.
\end{theorem}

As the static algorithm, we can use an $O(2^kk^9+|V||E|)$-time algorithm by
H{\"u}ffner, Komusiewicz, Moser, and Niedermeier~\cite{DBLP:journals/mst/HuffnerKMN10} 
or an $O(1.9102^k|G|)$-time algorithm by Boral, Cygan, Kociumaka, and Pilipczuk~\cite{DBLP:journals/corr/BoralCKP13}.

\subsection{Data Structure}
\begin{table}[tb]
\centering
\begin{tabular}{|c|l|}
\hline
$X$ & 3-approximate solution \\ 
\hline
$C_l \text{ for each cluster label } l$ & the vertices in the cluster labeled
$l$
\\
\hline
$l_u \text{ for each } u\in V\setminus X$ & label of the cluster that
$u$ belongs to \\
\hline
$L_x \text{ for each } x\in X$ & $\{l_u \mid u\in N(x)\setminus X\}$ \\
\hline
$P^+_{x,l} \text{ for each } x\in X \text{ and } l\in L_x$ & $C_l\cap N(x)$ \\
\hline
$P^-_{x,l} \text{ for each } x\in X \text{ and } l\in L_x$ & $C_l\setminus N(x)$ \\
\hline
\end{tabular}
\caption{Variables maintained in the algorithm}
\label{tbl:CVD}
\end{table}

\begin{figure}[tb]
 \begin{center}
  \includegraphics[height=40mm]{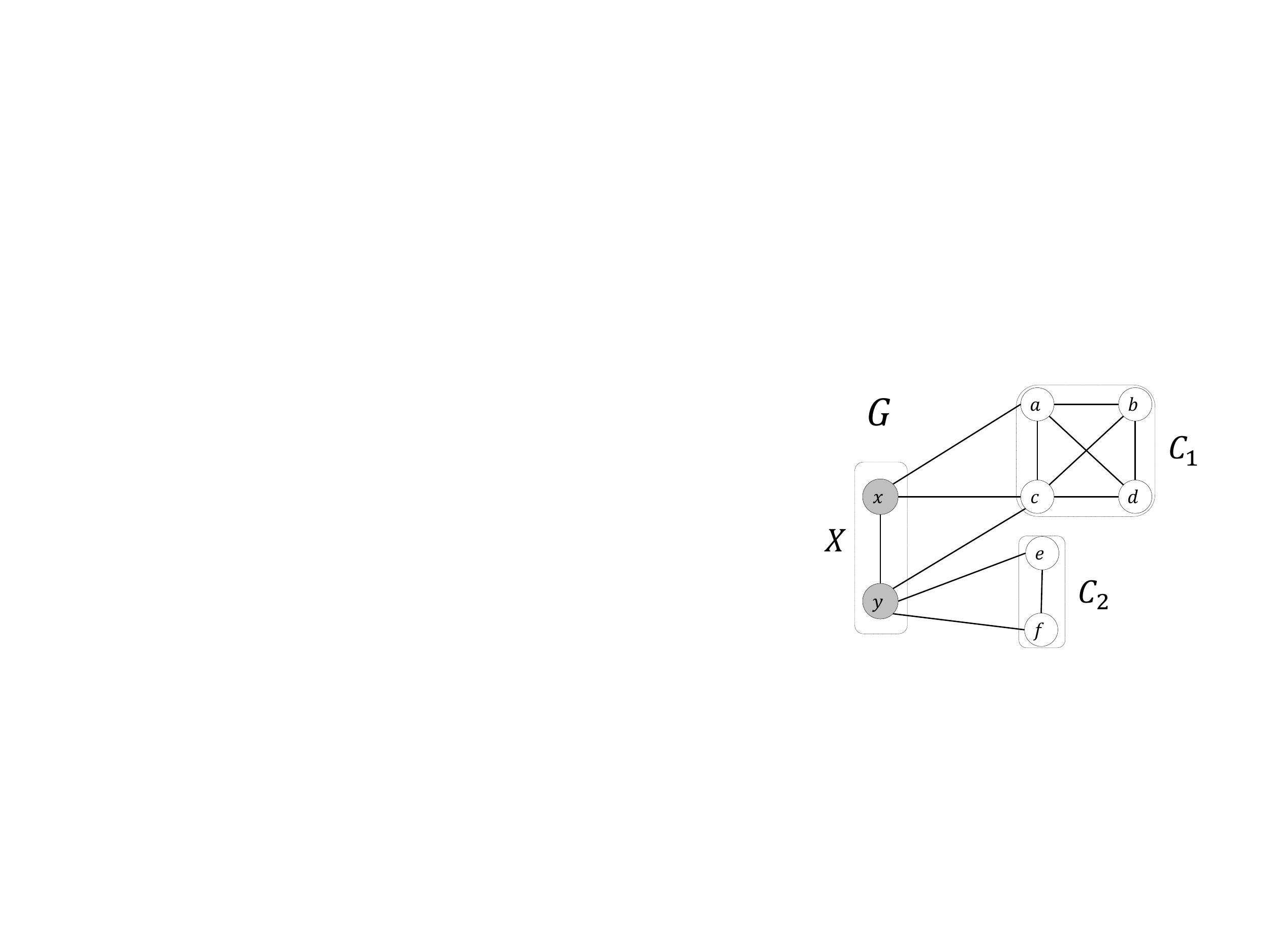}
 \end{center}
 \caption{An example of a graph and a 3-approximate solution.}
 \label{fig:CVD}
\end{figure}
We dynamically maintain the information listed in Table~\ref{tbl:CVD}.
We always keep a 3-approximate solution $X$.
Each cluster in $G[V\setminus X]$ is assigned a distinct \emph{cluster label}.
For each cluster label $l$, $C_l$ is the set of vertices on the cluster having the label $l$.
We keep the vertex set $C_l$ by using a persistent data structure that supports an
update in $O(\log |C_l|)$ time. One of such data structures is a persistent
red-black tree developed by Driscoll, Sarnak, Sleator and
Tarjan~\cite{DBLP:journals/jcss/DriscollSST89}.
The reason why the persistent data structure is
employed is that it enables us to copy the set in constant time.
For each $u\in V\setminus X$, $l_u$ is a label of the cluster $u$
belongs to.
For a vertex $x$ and a cluster, we say that $x$ is incident to the cluster if
at least one vertex in the cluster is incident to $x$.
For each $x\in X$, $L_x=\{l_u \mid u\in N(x)\setminus X\}$ is the 
labels of the clusters that $x$ is incident to.
For each $x\in X$ and $l\in L_x$, $P^+_{x,l} = C_l\cap N(x)$ is the set of the
neighbors of $x$ in $C_l$ and $P^-_{x,l} = C_l\setminus N(x)$ is the set of the
non-neighbors of $x$ in $C_l$. 
Note that all the variables are uniquely determined when $G$, $X$ and the
labels for all clusters are fixed.

For example, look at the graph depicted in Fig.~\ref{fig:CVD}.
$X=\{x,y\}$ is a 3-approximate solution, and $C_1$ and $C_2$ are clusters.
Here, the set of cluster labels is $\{1,2\}$, $l_a = l_b = l_c = l_d = 1$ and
 $l_e = l_f = 2$. $L_x = \{1\}$ and $L_y = \{1,2\}$. 
 $P^+_{x,1} = \{a,c\}$, $P^-_{x,1} = \{b,d\}$,
 $P^+_{y,1} = \{c\}$,   $P^-_{y,1} = \{a,b,d\}$,
 $P^+_{y,2} = \{e,f\}$ and $P^-_{y,2} = \{\}$.

\subsection{Algorithm}
\subsubsection{Update}
Let us explain how to update the data structure when an edge is added or
removed.
Before describing the whole algorithm, let us explain subroutines used in
the algorithm.
Algorithm~\ref{alg:CVD:add} is used to add a vertex $u$ in $V\setminus
X$ to $X$, and Algorithm~\ref{alg:CVD:remove} is to remove a
vertex $y$ from $X$ under the condition that $X\setminus\{y\}$ is still a
cluster vertex deletion.
Given a cluster vertex deletion $X$, Algorithm~\ref{alg:CVD:compress} computes
a 3-approximate solution $X'$.

\begin{algorithm}[tb] 
\caption{add $u\in V\setminus X$ to $X$}
\label{alg:CVD:add}
\begin{algorithmic}[1]
\State $l := l_u$ \label{lne:cvd:add:l}
\State remove $u$ from $C_l$ \label{lne:cvd:add:C}
\ForAll{any $x\in X$ such that $l\in L_x$} \label{lne:cvd:add:loop}
	\If{$\{x,u\}\in E$} \label{lne:cvd:add:ifE}
		\State remove $u$ from $P^+_{x,l}$ \label{lne:cvd:add:remPP}
		\State If $P^+_{x,l}$ becomes empty, remove $l$ from $L_x$
		\label{lne:cvd:add:remL}
	\Else
		\State remove $u$ from $P^-_{x,l}$ \label{lne:cvd:add:remPM}
	\EndIf
\EndFor
\State add $u$ to $X$ \label{lne:cvd:add:X}
\If{$C_l$ is still not empty} \label{lne:cvd:add:ifC}
	\State $L_u := \{l\}$  \label{lne:cvd:add:L}
	\State copy $C_l$ into $P^+_{u,l}$ \label{lne:cvd:add:copy}
	\State $P^-_{u,l} := \emptyset$ \label{lne:cvd:add:PMu}
\Else
	\State $L_u := \emptyset$
\EndIf \label{lne:cvd:add:endC}
\end{algorithmic}
\end{algorithm}

\begin{lemma}\label{lem:CVD:add}
Algorithm~\ref{alg:CVD:add} adds a vertex $u$ to $X$ and  
updates the data structure correctly in $O(|X|\log n)$ time.
\end{lemma}
\begin{proof}
At line~\ref{lne:cvd:add:l}, $l$ is the label of the cluster that the vertex $u$ belongs to.
By removing $u$ from $C_l$ at line~\ref{lne:cvd:add:C}, $C_l$ is correctly updated.

At line~\ref{lne:cvd:add:loop}, we iterate over all $x\in X$ such that $x$ is incident to the cluster $C_l$.
For the other vertices in $X$, since all clusters
except $C_l$ are not changed, we need no updates.
If there is an edge between $x$ and $u$, 
$u$ is in $P^+_{x,l}$.
Thus we remove $u$ from the set to correctly update
$P^+_{x,l}$ (line~\ref{lne:cvd:add:remPP}). If $P^+_{x,l}$ becomes empty by the
operation, it means that $C_l$ is no longer incident to $x$, and thus we remove $l$
from $L_x$ (line~\ref{lne:cvd:add:remL}).
If there is no edge between $x$ and $u$, $u$ is in $P^-_{x,l}$.
Thus we remove $u$ from the set to correctly update $P^-_{x,l}$
(line~\ref{lne:cvd:add:remPM}).
At line~\ref{lne:cvd:add:X}, we add $u$ to $X$ and complete the update of $X$.

Finally we initialize $L$, $P^+$ and $P^-$ for $u$ (line~\ref{lne:cvd:add:ifC} to~\ref{lne:cvd:add:endC}).
If $C_l$ is now empty, $u$ is not incident to any cluster, and thus $L_u$ is initialized as an empty set. Otherwise, $u$
is incident only to the cluster $C_l$. Thus we initialize $L_u$ as $\{l\}$
(line~\ref{lne:cvd:add:L}).
Since $C_l\cup \{u\}$ was a cluster, $u$ is incident to every vertex in $C_l$.
Thus we initialize $P^+_{x,l}$ copying $C_l$, and initialize $P^-_{x,l}$ to be
empty set (line~\ref{lne:cvd:add:copy} and~\ref{lne:cvd:add:PMu}).
Using a persistent data structure, copying $C_l$ into $P^+_{x,l}$ can be
done in $O(1)$ time.

Removing a vertex from $P^+_{x,l}$ or $P^-_{x,l}$ (line~\ref{lne:cvd:add:remPP}
and~\ref{lne:cvd:add:remPM}) takes $O(\log n)$ time, and this part is repeated
at most $|X|$ times.
This is the dominant part of the algorithm. 
Thus the time complexity of the algorithm is $O(|X|\log n)$.
\end{proof}

\begin{algorithm}[tb] 
\caption{remove $y\in X$ from $X$ assuming $G[V\setminus (X\setminus \{y\})]$ is
a cluster graph}
\label{alg:CVD:remove}
\begin{algorithmic}[1]
\State remove $y$ from $X$
\If {$L_y = \emptyset$} \label{lne:cvd:rem:IfLy}
	\State $l_y:=$ new label \label{lne:cvd:rem:}
	\State $C_{l_y}:=\{y\}$
\Else
	\State $|L_y|$ must be one. Let $l_y$ be the unique element in $L_y$.
	\State add $y$ to $C_{l_y}$
\EndIf \label{lne:cvd:rem:EndLy}
\State $l:=l_y$ \label{lne:cvd:rem:l}
\ForAll {$x\in X$ such that $l\in L_x$}
	\If {$\{x,y\}\in E$}
		{add $y$ to $P^+_{x,l}$}
	\Else
		{\enspace add $y$ to $P^-_{x,l}$}
	\EndIf
\EndFor
\ForAll {$x\in X$ such that $y\in N(x)$ and $l\notin L_x$}
	\State add $l$ to $L_x$
	\State $P^+_{x,l} := \{y\}$
	\State copy $C_l$ into $P^-_{x,l}$ and remove $y$ from $P^-_{x,l}$
	\label{lne:cvd:rem:copy}
\EndFor
\end{algorithmic}
\end{algorithm}
\begin{lemma}\label{lem:CVD:remove}
If $G[V\setminus (X\setminus \{y\})]$ is a cluster graph, 
Algorithm~\ref{alg:CVD:remove} removes a vertex $y$ from $X$ and
updates the data structure correctly in
$O(|X|\log n)$ time.
\end{lemma}
\begin{proof}
First, we remove $y$ from $X$ and complete the update of $X$.

From line~\ref{lne:cvd:rem:IfLy} to line~\ref{lne:cvd:rem:EndLy}, we compute
$l_y$, the label of the cluster that the vertex $y$ belongs to, and update $C_{l_y}$. 
Note that for any cluster label $l'\neq l_y$, $C_{l'}$ is not affected by the
removal of $y$.
If $L_y$ is empty, 
it means that there is not adjacent vertex of $y$ in $V\setminus X$, and thus
we create a new cluster label for the cluster $\{y\}$. Otherwise, from the
assumption that $G[V\setminus(X\setminus\{y\})]$ is a cluster graph, $y$ is adjacent to exactly one cluster. 
Let $l_y$ be the unique label in $L_y$, and then we add $y$ to $C_{l_y}$.

Let $l$ be $l_y$ for notational brevity (line~\ref{lne:cvd:rem:l}).
We update the values $L_x, P^+_{x,l}$ and $P^-_{x,l}$ for each $x\in X$.
Again, note that for any cluster label $l'\neq l$, $P^+_{x,l'}$ and $P^-_{x,l'}$
are not changed by the removal of $y$.
If $l\in L_x$, or equivalently $x$ is already adjacent to the cluster $C_l$
before $y$ is added to $C_l$, then we add $y$ to $P^+_{x,l}$ or $P^-_{x,l}$
according to whether $y$ is adjacent to $x$ or not.
If $l\notin L_x$, then $x$ is not adjacent to any vertex in $C_l$ before $y$ is
added to.
If $x$ is also not adjacent to $y$, no updates are needed.
If $x$ is adjacent to $y$, then $x$ is incident to the cluster $C_l$ after the
removal of $y$ from $X$. Thus we add $l$ to $L_x$, and initialize $P^+_{x,l}$ as
$\{y\}$ and $P^-_{x,l}$ as $C_l\setminus \{y\}$. 
To create $P^-_{x,l}$, we copy $C_l$ into $P^-_{x,l}$ and remove $y$ from
$P^-_{x,l}$. By using a persistent data structure, the update can be done in
$O(\log n)$ time.

It is easy to check that the time complexity of the algorithm is $O(|X|\log n)$.
\end{proof}

\begin{algorithm}[tb]
\caption{compute a new 3-approximate solution $X'$}
\label{alg:CVD:compress}
\begin{algorithmic}[1]
\State $V' := \emptyset$																	
\State $X_0 := \emptyset$																	
\ForAll {$x\in X$}																			
	\If {$|L_x| > |X| + 1$}																	
		\State add $x$ to $X_0$																
	\Else																					
		\State add $x$ to $V'$																
		\ForAll {$l\in L_x$}
			\State take $\min(|P^+_{x,l}|, |X| + 1)$ vertices from $P^+_{x,l}$, and
			add them to $V'$																
			\State take $\min(|P^-_{x,l}|, |X| + 1)$ vertices from $P^-_{x,l}$, and
			add them to $V'$																
		\EndFor																				
	\EndIf																					
\EndFor																					
\State $Y := $ 3-approximate cluster vertex deletion of $G[V']$ \label{lne:cvd:cmp:Y}	
\If{$|Y|>|X\setminus X_0|$}															%
	{$X' := X$}
\Else
	{\enspace $X' := X_0\cup Y$}																	
\EndIf
\end{algorithmic}
\end{algorithm}

\begin{lemma}\label{lem:CVD:compress}
Algorithm~\ref{alg:CVD:compress} computes a 3-approximate solution in $O(|X|^8)$ time.
\end{lemma}
In order to prove Lemma~\ref{lem:CVD:compress}, let us prove Lemma~\ref{lem:CVD_right}
and~\ref{lem:CVD_left}.

\begin{lemma}\label{lem:CVD_right}
Let $V'$ and $X_0$ be the sets computed by Algorithm~\ref{alg:CVD:compress}.
If $S\subseteq V'$ is a cluster vertex deletion of $G[V']$ such that $|S|\leq
|X\setminus X_0|$, then $S\cup X_0$ is a cluster vertex deletion of $G$.
\end{lemma}

\begin{proof}
Assume that $S$ is not a cluster vertex deletion of $G[V\setminus
X_0]$.
This implies that there is an induced $P_3$ in $G[(V\setminus X_0)\setminus
S]$. Let $x,y$ be vertices in $X\setminus X_0$ and $u,v$ be vertices
in $V\setminus(X\cup S)$.
There are four possible types of induced paths: (1) $xuy$, (2) $xyu$, (3) $xuv$,
and (4) $uxv$.
We will rule out all these cases by a case analysis
(see Fig.~\ref{fig:CVD_proof_right}). 
\begin{itemize}
\item[(1)] Let
$A=\{w\in V'\cap C_{l_u}\mid xw\in E\wedge yw\notin E\}$, $B=\{w\in V'\cap C_{l_u}\mid xw\notin E\wedge
yw\in E\}$ and $C=\{w\in V'\cap C_{l_u}\mid xw\in E\wedge yw\in E\}$.
By the construction of $V'$, $|A| + |C| \geq |X|+1$ and $|B| + |C| \geq |X| +
1$. Thus $\min\{|A|,|B|\}\geq |X| - |C| + 1$.
Since $x,y\notin S$ and $\{x,y\}\notin E$, $S$ must contain $C\cup B$ or
$C\cup A$.
Thus $|S| \geq |X| + 1$, which is a contradiction.
\item[(2)] Let $A=\{w\in V'\cap
C_{l_u}\mid xw\notin E\wedge yw\notin E\}$, $B=\{w\in V'\cap C_{l_u}\mid xw\in E\wedge
yw\in E\}$ and $C=\{w\in V'\cap C_{l_u}\mid xw\notin E\wedge yw\in E\}$.
By the construction of $V'$, $|A| + |C| \geq |X| + 1$ and $|B| + |C| \geq |X| +
1$. Thus $\min\{|A|,|B|\}\geq |X| - |C| + 1$.
Since $x,y\notin S$ and $\{x,y\}\in E$, $S$ must contain $C\cup B$ or $C\cup
A$. Thus $|S| \geq |X| + 1$, which is a contradiction.
\item[(3)] Since $|S| \leq |X|$, there
is a vertex $u'\in (V'\cap C_{l_u})\setminus S$ such that $\{x,u'\}\in E$ and a
vertex $v'\in (V'\cap C_{l_u})\setminus S$ such that $\{x,v'\}\notin E$. However
it contradicts the fact that $G[V'\setminus S]$ contains no induced $P_3$.
\item[(4)] Since $|S| \leq |X|$, there
is a vertex $u'\in (V'\cap C_{l_u})\setminus S$ such that $\{x,u'\}\in E$ and a
vertex $v'\in (V'\cap C_{l_v})\setminus S$ such that $\{x,v'\}\in E$. However
it contradicts the fact that $G[V'\setminus S]$ contains no induced $P_3$.
\end{itemize}
\end{proof}

\begin{figure}[tb]
 \begin{center}
  \includegraphics[height=35mm]{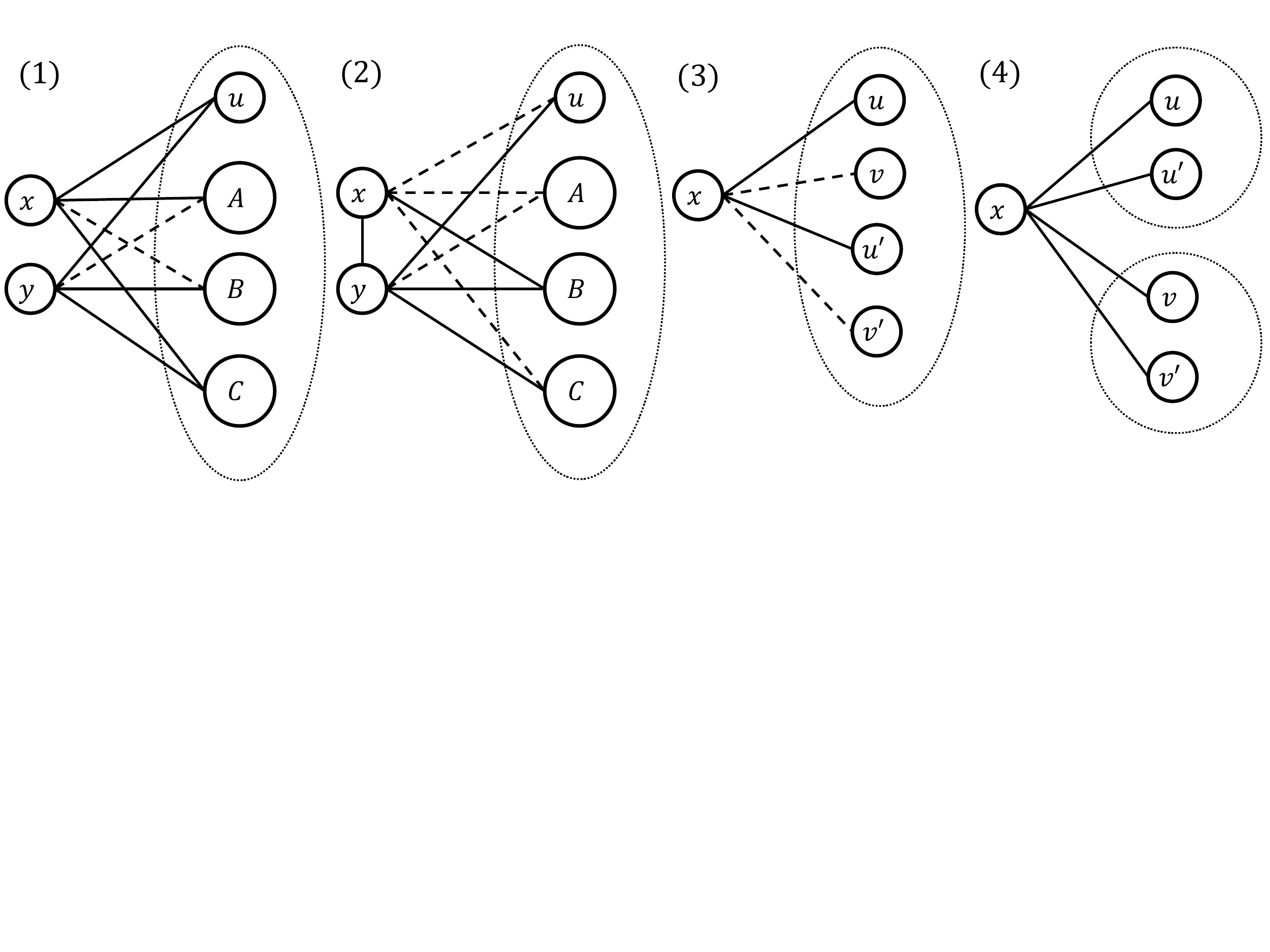}
 \end{center}
 \caption{Case analysis in the proof of Lemma~\ref{lem:CVD_right}. A dotted line
 denotes there is no edge(s)}
 \label{fig:CVD_proof_right}
\end{figure}

\begin{lemma}\label{lem:CVD_left}
Let $V'$ and $X_0$ be the sets computed by Algorithm~\ref{alg:CVD:compress}.
For any cluster vertex deletion $T$ of $G$ such that $|T|\leq |X|$, the following hold:
\begin{enumerate}
\item $T$ contains $X_0$,
\item $T\cap V'$ is a cluster vertex deletion of $G[V']$.
\end{enumerate}
\end{lemma}
\begin{proof}
First, let us prove that $T$ contains $X_0$. Assume there exists $x\in
X_0\setminus T$. Since $|L_x|$, the number of adjacent clusters of $x$, is more
than $|X|+1$, in order to avoid induced $P_3$, $T$ must contain at least $|L_x|-1 > |X|$
vertices from adjacent clusters. It contradicts the fact that $|T|\leq |X|$.
Thus, $T$ contains $X_0$, and so $T\setminus X_0$ is a cluster vertex deletion
of $G[V\setminus X_0]$.

Since $G[V\setminus T]$ is a cluster graph, its induced subgraph
$G[V'\setminus T]$ 
is also a cluster graph.
Thus $T\cap V'$ is a cluster vertex deletion of $G[V']$.
\end{proof}

\begin{proof}[Proof of Lemma~\ref{lem:CVD:compress}]
Let $X^*$ be a minimum cluster vertex deletion.
Since $X$ is a cluster vertex deletion, we have $|X^*| \leq |X|$.
By Lemma~\ref{lem:CVD_left}, it holds that $X_0\subseteq X^*$, and
$X^*\setminus X_0$ is a cluster vertex deletion of $G[V']$.
$X^*\setminus X_0$ is actually a minimum cluster vertex deletion of $G[V']$,
because otherwise there is a cluster vertex deletion $S$ of $G[V']$ such that
$|S|<|X^*\setminus X_0| \leq |X\setminus X_0|$, but then by
Lemma~\ref{lem:CVD_right}, $S\cup X_0$ becomes a cluster vertex deletion of
$G$ of size less than $|X^*|$, which is a contradiction.

If the size of the set $Y$ computed at line~\ref{lne:cvd:cmp:Y} is larger than $|X\setminus X_0|$, the set $X$ remains a
3-approximate solution.
Otherwise, from Lemma~\ref{lem:CVD_right}, $Y\cup X_0$ is a cluster vertex
deletion of $G$.
Since $Y$ is a 3-approximate solution and $X^*\setminus X_0$ is a minimum
cluster vertex deletion of $G[V']$, we have
\begin{eqnarray}
|Y\cup X_0|	& \leq 3|X^*\setminus X_0| + |X_0|
			& \leq 3|X^*|.
\end{eqnarray} 
Thus, $X' = Y\cup X_0$ is a 3-approximate solution on $G$.

The claimed time complexity is obtained as follows.
The size of $V'$ at line~\ref{lne:cvd:cmp:Y} is at most
$2|X|(|X|+1)^2 = O(|X|^3)$.
The number of edges in the graph $G[V']$ is maximized when $G[V'\setminus X]$ is
composed of $|X|+1$ cliques with size $|X|(|X|+1)$. Thus the number of edges
is at most $|X|^2(|X|+1)^3 = O(|X|^5)$.
Thus, a 3-approximate solution can be computed in $O(|X|^8)$ time
using the trivial algorithm described in
Section~\ref{sec:CVD_prob_time}, and thus the claimed time complexity holds.
\end{proof}

Now we are ready to describe how to update the data structure when an edge is
modified.
To add (remove) an edge $\{u,v\}$ to (from) a graph $G$, before modifying $G$, we add
$u$ and $v$ to $X$ one by one using
Algorithm~\ref{alg:CVD:add} unless the vertex is already in $X$.
After the operation, we add (remove) the edge $\{u,v\}\subseteq X$
to (from) $G$.
Note that we do not have to change our data structure by this operation.
Now $X$ is a cluster vertex deletion but may no longer be a 3-approximate solution.
Then we compute a new 3-approximate solution $X'$
using Algorithm~\ref{alg:CVD:compress}.

Finally we replace $X$ by $X'$ as follows.
Let $R$ be $X\setminus X'$ and $R'$ be $X'\setminus X$.
We begin with adding every vertex in $R'$ to $X$ one by one using
Algorithm~\ref{alg:CVD:add}. Then we remove every vertex in $R$ from $X$ one by
one using Algorithm~\ref{alg:CVD:remove}, and finish the replacement.
During the process, $X$ is always a cluster vertex deletion of the
graph, and thus the assumption of Algorithm~\ref{alg:CVD:remove} is satisfied.

Let $k$ be the maximum of the cvd numbers before and after the edge
modification.
During the above process, the size of $X$ is increased to at most $6k$.
Algorithm~\ref{alg:CVD:compress} is called only once, and
Algorithm~\ref{alg:CVD:add} and~\ref{alg:CVD:remove} are called $O(k)$ times.
Thus together with Lemma~\ref{lem:CVD:add},~\ref{lem:CVD:remove}
and~\ref{lem:CVD:compress}, the update time is $O(k^8 + k^2\log n)$.

\subsubsection{Query}
Let us explain how to answer a query.
To compute a minimum cluster vertex deletion $X'$, we use
almost the same algorithm as Algorithm~\ref{alg:CVD:compress}, but compute an exact
solution $Y$ at line~\ref{lne:cvd:cmp:Y} instead of an approximate solution.
The validity of the algorithm can be proved by almost the same argument.
The bottleneck of the algorithm is to compute a minimum cluster vertex deletion
of the graph $G[V']$.
Since the number of edges in $G[V']$ is $O(k^5)$ as noted in the proof of
Lemma~\ref{lem:CVD:compress}, using an $O(f(|G|,k))$-time static algorithm for \CVD,
we can obtain the solution in $O(f(k^5,k))$ time.
For example, using the algorithm in~\cite{DBLP:journals/tcs/ChenKX10}, we
can compute the solution in $O(1.9102^kk^5)$ time.

\section{Another Dynamic Graph for Cluster Vertex Deletion (constant time update and
query)}\label{sec:CVD2}
In this section, we give an $O(f(k))$-time dynamic graph for \CVD, where $k$ is
the cvd number and $f$ is some computable function.
This algorithm is not efficient than the algorithm in Section~\ref{sec:CVD} in
many cases in practice. However, we introduce this algorithm here because it
is asymptotically faster than the algorithm in Section~\ref{sec:CVD} if $k$ is a
constant, and is a base of the dynamic graph
for \CN parameterized by cvd number in Section~\ref{sec:CN}.

\subsection{Data Structure}
We maintain $X$ to be a minimum cluster vertex deletion. Initially
$X=\emptyset$.
Let us define an equivalence relation on $V\setminus X$ so
that two vertices $u,v\in V\setminus X$ are equivalent if and only if $u$ and
$v$ are in the same cluster and $N(u)\cap X = N(v)\cap X$.
The important point is that we do not have to distinguish the vertices in the
same class and can treat them as if they are one vertex weighted by the number
of the vertices in the same class. To treat the vertices in a same class efficiently, we introduce an
auxiliary data structure, an undirected graph $H$.
The graph $H$ is uniquely determined by $G$ and $X$ as follows (see
Fig.~\ref{fig:cluster}.)
\begin{figure}[tb]
 \begin{center}
  \includegraphics[height=25mm]{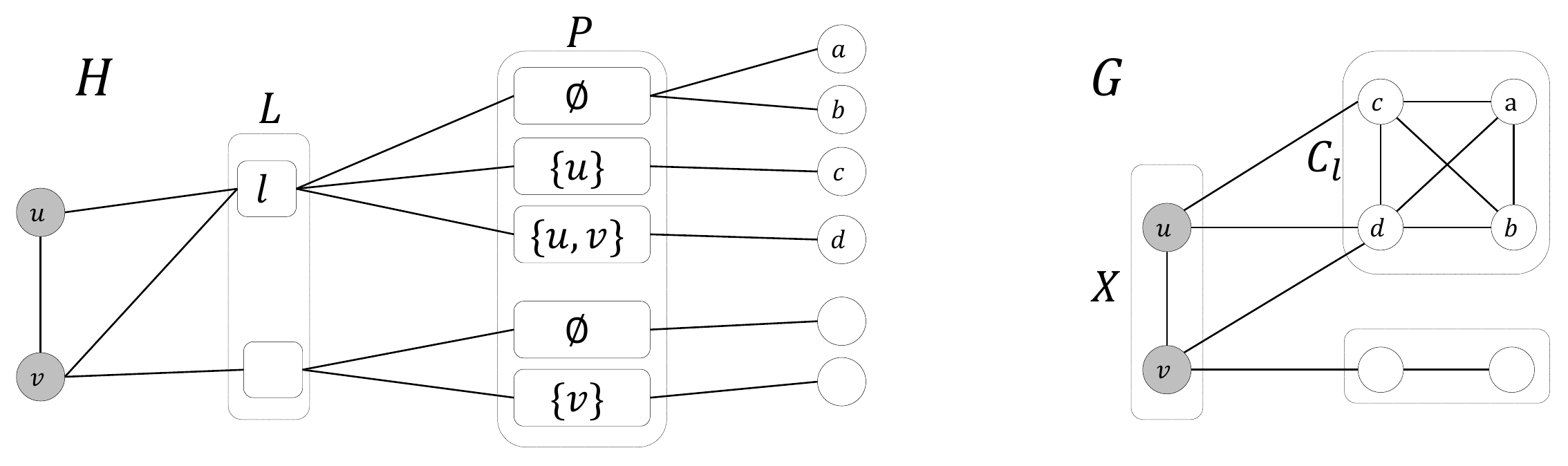}
 \end{center}
 \caption{Correspondence of $G$ and $X$ to $H$.}
 \label{fig:cluster}
\end{figure}

First, we introduce \emph{cluster labels} $L$ by assigning a different label $l\in L$ to a different cluster $C_l \subseteq V\setminus X$.
For each cluster label $l\in L$ and a vertex set $S\subseteq X$,
let $C_{l,S}$ denote the equivalent class $\{v\in C_l\mid N(v)\cap X = S\}$.
We introduce a \emph{class label} $p_{l,S}$ for each nonempty equivalent class
$C_{l,S}$. Let $P$ be the set of the introduced class labels.
The vertex set of $H$ is $X\cup L\cup P\cup (V\setminus X)$.

Then, we add an edge between $x$ and $y$ in $H$ if and only if one of the following
conditions is satisfied:
\begin{itemize}
  \item $x,y\in X$ and $e\in E(G[X])$,
  \item $x\in X, y\in L$ and $x\in N(C_y)$,
  \item $x\in L$ and $y=p_{x,S}\in P$ for some $S\subseteq X$, or
  \item $x=p_{l,S}\in P$ and $y\in C_{l,S}$.
\end{itemize}

\subsection{Strategy}\label{sec:cvd2:strategy}
Whenever the graph $G$ is updated, we transform the current
solution $X$ to a minimum solution as follows:
\begin{itemize}
\item[(1)] $X$ is extended to be a solution.
\item[(2)] A problem kernel $G[V']$ is obtained using the solution.
\item[(3)] A minimum solution $X'$ is obtained applying a (static)
algorithm for $G[V']$, and $X$ is exchanged for the minimum solution $X'$.
\end{itemize}
Let us call the steps (1), (2) and (3) as \emph{growing phase}, \emph{compression
phase} and \emph{exchange phase} respectively.
To compute an exact solution, we can apply any exact FPT algorithm to the current
problem kernel $G[V']$.

\subsubsection{Move a vertex to $X$}\label{sec:addCVD}
Let us show how to move a vertex $v$ in $V\setminus X$ to $X$ and update the graph $H$ efficiently.
This procedure is used in the growing phase and the exchange phase.

\begin{enumerate}
  \item[(1)] Let $C_{l,S}$ be the equivalent class containing $v$. Remove the
  edge $\{p_{l,S},v\}$ of $H$. If $C_{l,S}$ becomes empty, remove the label $p_{l,S}$
  destroying the incident edge to $l$. Moreover, if $C_l$ becomes empty, remove
  the label $l$ destroying all incident edges to $X$.
  \item[(2)] If the label $l$ is not removed, add the edge $\{v,l\}$. For each neighbor $u\in N_G(v)\cap
  X$, add the edge $\{u,v\}$. If $C_{l,S}$ becomes empty but $C_l$ is still not empty,
  for each $u\in N_G(v)\cap X$ that is no more adjacent to the
  cluster $C_l$, remove the edge between $u$ and $l$. Then, for each class
  label $p_{l,S'}$ adjacent to $l$, change its name to $p_{l,S'\cup\{v\}}$.
\end{enumerate}

Step (1) corresponds to the removal of $v$ from $V\setminus X$ and step (2)
corresponds to the addition of $v$ to $X$.
It is easy to see that this procedure correctly updates $H$.
Since there are at most $O(2^{|X|})$ equivalent classes for each cluster, 
the running time is $O(2^{|X|}|X|^2)$.

\subsubsection{Growing Phase}\label{sec:growCVD}
To handle an insertion or a deletion of an edge, we first move its endpoints
not in $X$ to $X$ to make $X$ a cluster vertex deletion.
This phase is completed in $O(2^kk^2)$ time, and now $|X|$ is at most $k+2$.

\subsubsection{Compression Phase}\label{sec:compressCVD}
If we remove a vertex $v\in X$ from $X$ that is adjacent to $d$ different clusters,
we have to include at least $d-1$ vertices into $X$ because there are no edges between
different clusters.
Thus, to improve the solution $X$, we can only remove a vertex $v\in X$ that is adjacent
to at most $|X|$ different clusters.
Let $X_0$ be the vertices in $X$ such that the number of its
adjacent clusters is at most $|X|$. As noted above, $X_1 = X\setminus X_0$ must
be kept in $X$ to improve the solution.
Let $L'$ be the labels of clusters adjacent to any vertex in
$X_0$. $|L'| \leq |X|^2$ by the definition of $X_0$.

Now we make a vertex weighted graph $G'$.
Let $P_l$ denote the set of the class labels of the form $p_{l,*}$.
The vertex set of $G'$ is $X_0\cup P'$, where $P'=\bigcup_{l\in L'}P_l$.
Since the size of $P_l$ is at most $2^{|X|}$, it holds $|V(H)|\leq |X| + 2^{|X|}|X|^2$. The weight $w:V(H)\rightarrow
\mathbb{N}$ is defined by $w(v) = 1$ for every $v\in X_0$ and $w(p_{l,S}) =
|C_{l,S}|$ for every $p_{l,S}\in P'$.
For each edge $e\in E(G[X_0])$, we add $e$ to $G'$.
For each $v\in X_0$ and $p_{l,S}\in P'$ we add the edge between them if $v\in
S$.
Finally, for each $l\in L'$, we add edges between every two class labels in
$P_l$, making $G'[P_l]$ a clique.
This completes the construction of the graph $G'$.

Then we solve \CVD for the vertex weighted graph $G'$, and determine
$X'$ as the vertices corresponding to the solution. Since $|V(H)|
\leq 2^kk^2 + k$, apparently it is solvable in $O(f(k))$ time for some function
$f$.
For example, using the $O(2^kk^9+|V||E|)$ algorithm proposed by
H{\"u}ffner, Komusiewicz, Moser and
Niedermeier~\cite{DBLP:journals/mst/HuffnerKMN10}, we can solve the problem in
$O(2^kk^9+8^kk^6) = O(8^kk^6)$ time.

\subsubsection{Exchange Phase}
Let $R$ be $X\cap X'$. To exchange $X$ for $X'$, we add all vertices in $X'\setminus R$ to
$X$ and then remove all vertices in $X\setminus R$ from $X$. Note that during the
process, $X$ is always a cluster vertex deletion. 
The additions to $X$ are executed using the method in Section~\ref{sec:addCVD}
at most $k$ times.

Let us show how to remove a vertex $v$ from $X$.
Before removing $v$ from $X$, we update $H$ as follows.
\begin{enumerate}
  \item[(1)] Since $X\setminus\{v\}$ is a cluster vertex deletion, $v$ is adjacent to at most one cluster label. If
  there is such a cluster label $l$, remove the edge $\{v,l\}$. Otherwise,
  introduce a new cluster label $l$ corresponding to an empty cluster.
  Let $S$ be $N_H(v)$, that is a subset of $X\setminus \{v\}$.
  Remove all edges from $v$ to $S$.
  \item[(2)]
  Add the edges from the vertex in $S$ to $l$ unless it already exists.
  Then, add the edges $\{l,p_{l,S}\}$ and $\{p_{l,S},v\}$, introducing the
  class label $p_{l,S}$ unless it already exists.
  Finally, we iterate over each class label $p_{l,S'}$ and change
  its name to $p_{l,S'\setminus \{v\}}$.
  Note that $v$ is in $S'$ by the assumption.
\end{enumerate}

Step (1) corresponds to the removal of $v$ from $X$ and step (2)
corresponds to the addition of $v$ to $X\setminus V$. It is easy to see
that this procedure takes $O(2^{|X|}|X|^2)$ time and correctly updates $H$.

During the exchange, the size of $X$ is increased to at
most $2k$. Thus the exchange phase can be done in $\sum_{k'=k+1}^{2k}O(k'^2 2^{k'})=O(k^22^{2k})$ time.

In summary, we have obtained an $O(8^kk^6)$-time dynamic graph algorithm for
\CVD. The dominant part is the compression phase.

\section{Dynamic Graph for Chromatic Number Parameterized by CVD Number}\label{sec:CN}
Let $G=(V,E)$ be a graph.
Given a vertex set $S\subseteq V$, $\Pi(S)$ denotes the set of all
\emph{partitions} of $S$.
The size of all partitions $|\Pi(S)|$ is called \emph{Bell number} $B_{|S|}$ and
it is known that $B_n = 
O((\frac{0.792n}{\mathrm{ln}(n + 1)})^n)$~\cite{berend2010improved}.
Note that we can easily enumerate all partitions in $\Pi(S)$ in $O(B_{|S|})$
time.
A partition of the vertices of a graph can be regarded as a \emph{coloring} of
the graph, since given a partition, by coloring vertices in the same block with
the same color, we can determine the corresponding coloring on the graph.
For $p\in \Pi(S)$ and $S'\subseteq S$, let $p|_{S'}$ denote the restriction of
$p$ into $S'$, that is the unique coloring $p'\in \Pi(S')$ that can be extended to $p$.
For a vertex set $S\subseteq V$, we call a
partition $p\in \Pi(S)$ a \emph{proper coloring} or \emph{proper} if no
two vertices in $S$ sharing the same edge are in a same block.
In other words, considering the corresponding coloring, if there are no adjacent
vertices with the same color, then the partition is a proper coloring. For a
coloring $p$, its \emph{size} $|p|$ is defined by the number of blocks of $p$.
\CN is a problem of finding the size of a minimum proper coloring of $V$. The
size is called \emph{chromatic number}.

In this section, we provide a data structure for {\sf Chromatic Number}, that
works efficiently when the size of a minimum cluster vertex deletion (cvd
number) of the graph is small.
\begin{theorem}\label{thm:CN_CVD}
There is a data structure treating the graph $G$ that supports the following
operations:
\begin{itemize}
\item[(1)] Compute the solution of {\sf Chromatic Number} of the current graph
$G$.
\item[(2)] Add an edge to $G$.
\item[(3)] Remove an edge from $G$.
\end{itemize}
The operation (1) takes $O(1)$ time, and both (2) and (3) take
$O(\log n)$ time assuming the cvd number is a constant.
More precisely, the time complexity of (1) is $O(1)$, and of (2) and (3) are
$O(B_{2k}(4^kk^3 + \log n))$, where $k$ is maximum of the cvd number of the
current graph and cvd number of the updated graph.
\end{theorem}
The first static FPT algorithm for this problem is proposed by Martin Doucha and Jan Kratochv\'il~\cite{DBLP:conf/mfcs/DouchaK12}. Our algorithm is based on the algorithm, but updates necessary information dynamically.

\subsection{Data Structure}
As an underlying data structure, we maintain the data structure described in
Section~\ref{sec:CVD2}.
For a cluster label $l$, let $X_l$ be $\{x\in X\mid N(x)\cap C_l\neq
\emptyset\}$. $X_l$ is easily computed in $O(k)$ time from the auxiliary data
structure $H$.
In addition, we maintain the following information:
\begin{itemize}
  \item For every cluster label $l$ and coloring $p\in \Pi(X_l)$, $\chi_{l,p}$
  is the size of a minimum proper coloring on $G[X_l\cup C_l] - E(G[X_l])$ made
  by extending $p$.
  \item For every $Y\subseteq X$ and $p\in \Pi(Y)$, $\Lambda_p$ is the multiset
  containing every $\chi_{l,p}$ such that $X_l = Y$. We use a balanced binary
  search tree to hold $\Lambda_p$ to efficiently remove or add a value and get
  the maximum value in the set.
\end{itemize}


More formally, $\chi_{l,p}$ is defined by
\begin{eqnarray}
\min\{|q| \mid & q & \in\Pi(C_l\cup X_l) \wedge q|_{X_l}=p \nonumber \\
& \wedge & q \text{ is a proper coloring on } G[X_l\cup C_l] - E(G[X_l])\},
\label{eqn:CN:chi}
\end{eqnarray}
and for $Y\subseteq X$ and $p\in \Pi(Y)$, $\Lambda_{p}$ is the multiset
\begin{eqnarray}
\{\chi_{l,p} \mid l \text{ is a cluster label } \wedge X_l=Y\}.
\label{eqn:CN:lam}
\end{eqnarray}

When all $\Lambda_p$ are given, we can compute the chromatic number of the
graph as:
\begin{eqnarray}
\min_{\text{proper coloring } p\in\Pi(X) \text{ on } G[X]}
\max\{|p|,
(\max_{Y\subseteq X}\text{ maximum value in }\Lambda_{p|_Y})
\}, \label{eqn:CN:sol}
\end{eqnarray}
where maximum value of an empty set is defined to be 0.

\begin{proof}
\begin{eqnarray}
\max\{|p|, (\max_{Y\subseteq X} (\max \Lambda_{p|_Y} ))\} \label{eqn:CN:sol2}
\end{eqnarray}
is the minimum number of colors needed to color the graph $G-E(G(X))$, where the
coloring on $X$ is fixed to be $p$.
Thus moving $p$ over all proper coloring on $G[X]$ and taking the minimum of
(\ref{eqn:CN:sol2}), we obtain the size of a minimum proper coloring.
\end{proof}

Since there are at most $n=|V|$ clusters, the maximum value in $\Lambda_{p|_Y}$
is obtained in $O(\log n)$ time. There are at most $O(B_{|X|}|X|)$
possibilities of $p|_Y$ in the formula~(\ref{eqn:CN:sol}). Thus the formula is
computed in $O(B_{|X|}|X|\log n)$ time.

\subsection{Algorithm to Update Data Structure}
Let us explain how to update the data structure when an edge is added or
removed.
Before describing the whole algorithm, let us introduce subroutines used in the
algorithm. 

\subsubsection{Add (Remove) a Vertex to (from) $X$.}
Let us explain how to add (remove) a vertex to (from) $X$ and
update the data structure accordingly. We
assume that when a vertex $y$ is removed from $X$, $G[V\setminus
(X\setminus\{y\})]$ is a cluster graph.

To add a vertex $u$ to $X$:
\begin{enumerate}
  \item Let $l$ be the label of the cluster that the vertex $u$ belongs to.
  For every $p\in \Pi(X_l)$, we remove $\chi_{l,p}$ from $\Lambda_p$ to prepare
  the update of the value $\chi_{l,p}$.
  \item We run the algorithm in Section~\ref{sec:CVD2}.
  In particular, $X$, $X_l$, $C_l$ and $C_{l,S}$ are updated for every
  $S\subseteq X$.
  \item If the cluster $C_l$ still exists, we compute $\chi_{l,p}$ for every
  $p\in \Pi(X_l)$ using Lemma~\ref{lem:CN_recompute}, and then add $\chi_{l,p}$
  to $\Lambda_p$.
\end{enumerate}

To remove a vertex $y$ from $X$:
\begin{enumerate}
  \item From the assumption, $y$ is adjacent to at most one cluster. If there is
  a cluster that $y$ is adjacent to, let $l$ be the label of the cluster, and we
  remove $\chi_{l,p}$ from $\Lambda_p$ to prepare the update of the value
  $\chi_{l,p}$.
  \item We run the algorithm in Section~\ref{sec:CVD2}. 
  In particular, $X$, $X_l$, $C_l$ and $C_{l,S}$ are updated for every
  $S\subseteq X$, where $l$ is the label of the cluster that $y$ belongs to now.
  \item We compute $\chi_{l,p}$ for every $p\in \Pi(X_l)$ using
  Lemma~\ref{lem:CN_recompute}, and then add $\chi_{l,p}$ to $\Lambda_p$.
\end{enumerate}

Now, we prove that $\chi_{l,p}$ can be computed efficiently.
\begin{lemma}\label{lem:CN_recompute}
After the equivalent class $C_{l,S}$ is updated for each $S\subseteq
X$, for any $p\in\Pi(X_l)$, $\chi_{l,p}$ can be computed in $O(2^{|X|}|X|^3)$
time.
\end{lemma}
\begin{proof}
We want to find the size of a minimum proper coloring $q$ on $G[X_l\cup C_l] -
E(G[X_l])$ such that $q|_{X_l} = p$.
Without loss of generality, let the set of colors used for $X_l$ be
$\{1,\ldots,|p|\}$. For $S\subseteq X_l$, let $c(S)\subseteq \{1,\ldots,|p|\}$
denote the colors assigned to $S$.

Since $G[C_l]$ is a cluster, the colors assigned to $C_l$ must be distinct.
Minimizing the size of the coloring is equivalent to maximizing the number of
vertices in $C_l$ that are colored with 1,\ldots,$|p|$. 

Let $r$ be the number of vertices in $C_l$ that are assigned a color in
$\{1,\ldots,|p|\}$. We want to maximize $r$ to minimize the number of
colors $|p|+|C_l| - r$.

We compute the maximum possible $r$ by constructing a graph $(\{s\}\cup L\cup
R\cup \{t\}, F)$ and computing the size of a maximum flow from $s$ to $t$.
Create vertices $x_1,\ldots,x_{|p|}$ and let $L$ be
$\{x_1,\ldots,x_{|p|} \}$.
We add edges $\{s, x_1\},\ldots,\{s, x_{|p|}\}$ with capacity one.
For each $S\subseteq X$, we create a vertex $y_S$ and add $y_S$ to $R$, and add
the edge $\{y_S,t\}$ with capacity $|C_{l,S}|$.
Then, for each $S\subseteq X$ and $i\in \{1,\ldots,|p|\}\setminus c(S)$,
we add an edge between $x_i$ and $y_S$ with capacity one, completing the
construction of the graph.
We compute $r$ as the size of a maximum $s$-$t$ flow of the constructed graph,
and conclude $\chi_{l,p}$, the size of a minimum coloring on $G[X_l\cup
C_l]-E(X_l)$ obtained by extending $p$, is $|p| + |C_l| - r$.

The size of $L$ is at most $|X|$ and the size of $R$
is at most $2^{|X|}|X|$. Thus the number of edges in the graph is at most
$O(2^{|X|}|X|^2)$.
Since the size of a maximum flow is at most $|L|\leq |X|$, using Ford-Fulkerson
algorithm~\cite{ford1956maximal}\footnote{whose time complexity is
$O((\text{\em the number of edges})(\text{\em maximum flow size}))$}, we compute
the solution in $O(|F||X|)=O(2^{|X|}|X|^3)$ time.
\end{proof}

Since $u$ was not adjacent to any cluster other than $C_l$, no change of
$\chi_{l',p'}$ is needed for any cluster label $l'\neq l$ and $p'\in
\Pi(X_{l'})$. Thus we have obtained the following lemma.
\begin{lemma}\label{lem:CN:time_add}
We can add (remove) a vertex to (from) $X$ and
update the data structure accordingly in
$O(B_{|X|}\log n + B_{|X|}2^{|X|}|X|^3)$ time.
\end{lemma}
$O(B_{|X|}\log n)$ is the time to remove (add) $\chi_{l,p}$ from (to)
$\Lambda_p$ for every $p\in\Pi(X_l)$, and $O(B_{|X|}2^{|X|}|X|^3)$ is the time
to compute $\chi_{l,p}$ for every $p\in \Pi(X_l)$.

\subsubsection{Add or Remove an Edge}
Now we are ready to describe how to update the data structure when the graph is
modified.

As we did in Section~\ref{sec:CVD2},
before the edge $\{u,v\}$ is added to or removed from $G$, we move $u$ and $v$
to $X$ making $X$ a cluster vertex deletion
of the modified graph.
When the edge is actually added to (removed from) $G$, 
no change of the data structure is needed since any edge lying on $X$ does not
affect $\chi_{l,p}$ for any cluster label $l$ and $p\in\Pi(X_l)$.

After the edge is added to (removed from)$G$, we compute a minimum
cluster vertex deletion $X'$ using the algorithm in Section~\ref{sec:CVD2}. Then we replace $X$ with $X'$.
Let $R$ be $X\setminus X'$ and $R'$ be $X'\setminus X$.
We begin with adding every vertex in $R'$ to $X$ one by one. Then we remove
every vertex in $R$ from $X$ one by one finishing
the replacement.
During the process, $X$ is always a cluster vertex deletion of the
graph, and thus the algorithm works correctly.
As the size of $X$ is increased to at most $2k$,
by Lemma~\ref{lem:CN:time_add}, the replacement is completed in
$\sum_{k'=k}^{2k}O(B_{k'}\log n + B_{k'}2^{k'}k'^3) = O(B_{2k}(\log n +
4^kk^3))$ time.

After all information is updated, finally we compute the solution of \CN
using the formula (\ref{eqn:CN:sol}). As noted above, it is
computed in $O(B_{|X|}|X|\log n) = O(B_kk\log n) = O(B_{2k}\log n)$ time.

Putting things altogether, now we have proved
Theorem~\ref{thm:CN_CVD}.

\section{Dynamic Graph for Bounded-Degree Feedback Vertex Set}\label{sec:feedback}
A vertex set is called \emph{feedback vertex set} if its removal makes the graph
a forest.
\FVS is the problem of finding a minimum feedback vertex set.
In this section, we assume every vertex in $G$ always has degree at most $d$.
We maintain $X$ to be a minimum feedback vertex set. Initially $X=\emptyset$.
As noted in the Introduction, our algorithm is based on the static algorithm by 
Guo, Gramm, H{\"u}ffner, Niedermeier and
Wernicke~\cite{DBLP:journals/jcss/GuoGHNW06}.
\subsection{Data Structure}
The key point is that we keep the forest $G[V\setminus X]$ using dynamic
tree data structures called \emph{link-cut tree} data structures presented by
Sleator and Tarjan~\cite{DBLP:journals/jcss/SleatorT83}.
Link-cut tree is a classic data structure that supports many operations on a
forest in $O(\log n)$ amortized time.
We exploit the following operations of link-cut trees. All of them are
amortized $O(\log n)$-time operations:
\begin{itemize}
  \item $\mathbf{link}(r,v)$: If the vertex $r$ is a root of a tree and the
  vertices $r$ and $v$ are in different trees, add an edge from $r$
  to $v$.
  \item $\mathbf{cut}(u,v)$: Remove the edge between $u$ and $v$.
  \item $\mathbf{evert}(v)$: Make the vertex $v$ a root of the tree containing
  $v$.
  \item $\mathbf{root}(v)$: Return the root of the tree containing $v$.
  \item $\mathbf{nca}(u,v)$: If the vertices $u$ and $v$ are in the same tree,
  find the nearest common ancestor of $u$ and $v$ in the rooted tree containing $u$ and
  $v$.
  \item $\mathbf{parent}(v)$: If $v$ is not a root, return the parent of
  $v$.
\end{itemize}
For further understanding, see~\cite{DBLP:journals/jcss/SleatorT83}.
Note that we can check whether the vertices $v$ and $u$ are in the same tree
testing if $\mathbf{root}(v)$ equals $\mathbf{root}(u)$ or not.

\subsection{Strategy}\label{sec:fvs:strategy}
As with Section~\ref{sec:cvd2:strategy}, whenever the graph $G$ is updated, we
transform the current solution $X$ to a minimum solution as follows:
\begin{itemize}
\item[(1)] $X$ is extended to be a solution.
\item[(2)] A problem kernel $G[V']$ is obtained using the solution.
\item[(3)] A minimum solution $X'$ is obtained applying a (static)
algorithm for $G[V']$, and $X$ is exchanged for the minimum solution $X'$.
\end{itemize}
Let us call the steps (1), (2) and (3) \emph{growing phase}, \emph{compression
phase} and \emph{exchange phase} respectively.
To compute an exact solution, we can apply any exact FPT algorithm for current
problem kernel $G[V']$.

\subsection{Growing Phase}\label{sec:fvsGrowing}
When an edge $e=\{u,v\}$ is added to $G$, we insert one of its endpoints, say
$u$, to current solution. That is, we add $u$ to $X$ and remove $u$ from the
link-cut tree containing $u$ calling $\mathbf{cut}(u,v)$
for each $v$ in $N(u)\setminus X$. The time complexity of this phase is
$O(d\log n)$. When an edge is removed from $G$, we do nothing.
After the operation, $X$ is a feedback vertex set and the size of $X$ is at most
$k+1$.
\subsection{Compression Phase - Outer Loop}
In our algorithm, we only consider special kinds of solutions of \FVS. We call
them \textit{maximum overlap solutions}.
\begin{definition}[Maximum Overlap Solution]
For given graph $G$ and a feedback vertex set $X$ of $G$,
$X'\subseteq V$ is called \textit{maximum overlap solution} if $X'$ is a minimum feedback
vertex set of $G$ and the size of intersection $|X\cap X'|$ is the maximum
possible.
\end{definition}
At least one maximum overlap solution apparently always exists.
Then, we try to find a maximum overlap solution $X'$.
To do this, we use the idea of exhaustively considering
all possible intersections of $X\cap X'$. That is, we compute a minimum feedback
vertex set $X''$ that satisfies $X\cap X'' = R$ for all $2^{|X|}$ possibilities
of $R\subseteq X$, and let $X'$ be the minimum of them.
Now our job is to solve the following task.
\begin{task}[\textsf{Disjoint Feedback Vertex Set}\xspace]\label{task:disjFVS}
Given a feedback vertex set $X$ and a subset $R$ of $X$, find a minimum vertex
set $S'$ such that $S'$ is a feedback vertex set of $G[V\setminus R]$ and
$S\cap S'= \emptyset$, where $S=X\setminus R$, or output `NO' if there is no
such $S'$.
\end{task}
The important point is that for the purpose of solving \FVS, 
it is sufficient to
obtain the algorithm such that (1) if the algorithm outputs a vertex set $S'$, it is a
correct answer and (2) if there is a maximum overlap solution $X'$ such that
$X\cap X'=R$, the algorithm outputs a correct answer, i.e.,
the algorithm can incorrectly output `NO' if there is no maximum overlap
solution $X'$ such that $X\cap X'=R$.
Thus we consider such an algorithm.

\subsection{Compression Phase - Inner Loop in $O(f(k)|E|)$ time}
First, let us introduce an $O(f(k)|E|)$ algorithm.
Let $V'$ be $V\setminus R$.
We reduce the graph $G[V']$ applying the following two rules recursively.
\begin{enumerate}
  \item If there is a degree 1 vertex $v\in V'\setminus S$, remove $v$ from the
  graph
  \item If there is a degree 2 vertex $v\in V'\setminus S$, remove $v$ and connect its
  neighbors by an edge. It may make parallel edges between the neighbors.
\end{enumerate}
Note that after the reduction every vertex not in $S$ has degree at least
three.
Let $\mathbf{reduced}(R)$ denote the (uniquely determined) reduced graph (see
Fig.~\ref{fig:fvsCore_Reduced}).
  
\begin{figure}[htbp]
 \begin{center}
  \includegraphics[height=25mm]{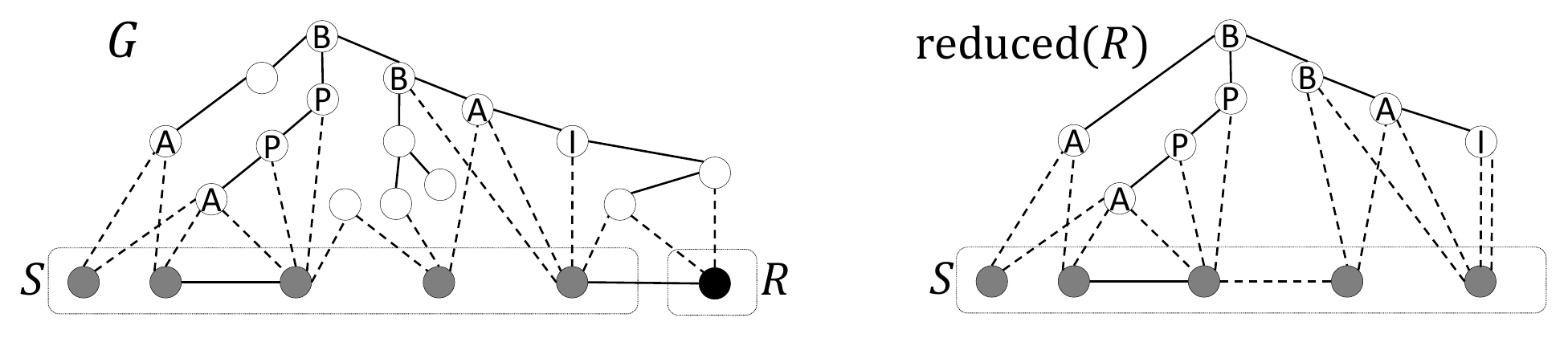}
 \end{center}  
 \caption{ 
 Example of the original graph $G$ and the reduced graph.
 The white vertices with some letters are core vertices, and each letter denotes
 the class of the vertex defined in the proof of Lemma~\ref{lem:boundFVS}.
 }
 \label{fig:fvsCore_Reduced} 
\end{figure}

After the reduction, we try to find a minimum feedback vertex
set disjoint from $S$ in the reduced graph.

This reduction is based on the iterative compression algorithm in~\cite{DBLP:journals/jcss/GuoGHNW06}.
The differences are: (1) we do not reduce the graph even if there are
parallel edges, and (2) we contract the vertex $v$ even if $N(v)\subseteq S$.
The second difference may lead to the algorithm that incorrectly outputs `NO'.
However, by the following lemma, such cases never happen when $R$ is a correct
assumption.
\begin{lemma}\label{lem:niceFVS}
In the setting of Task~\ref{task:disjFVS}, if there is a maximum overlap solution $X'$ such
that $R=X\cap X'$, $\mathbf{reduced}(R)$ has a feedback
vertex set $S'$ such that $S'\cap S = \emptyset$ and $|S'| + |R| =
|X'|$.
\end{lemma}
\begin{proof}
We prove that any vertex in $X'\setminus R$ never removed during the
reduction, and the lemma follows since $X'\setminus R$ becomes a minimum
feedback vertex set of $\mathbf{reduced}(R)$.
Let us hypothesize that during the reduction, a vertex $v$ in
$X'\setminus R$ has been removed. It means that there are at most two edge
disjoint paths from $v$ to vertices in $S$. If there is no path from $v$ to $S$,
$X'\setminus \{v\}$ must be a feedback vertex set of $G$ and contradict the
minimality of $X'$.
Otherwise let $u\in S$ be the end point of a path from $v$ to $S$.
Then $X'\setminus \{v\} \cup \{u\}$ must be also a minimum feedback vertex set
and contradict the maximality of $|X\cap X'|$.
\end{proof}
Thus, to solve \FVS for the original graph, we only have to solve \FVS for the
reduced graph. In fact, if the reduced graph has a smaller solution, the graph
must be not so large.
Let us call the vertex set $V(G')\setminus S$ \emph{core vertices}.
As in \cite{DBLP:journals/jcss/GuoGHNW06}, we can prove the following lemma.
\begin{lemma}\label{lem:boundFVS}
In the setting of Task~\ref{task:disjFVS}, if there is a maximum overlap solution $X'$ such
that $R=X\cap X'$, the size of the core vertices is no more than $13|S|$.
\end{lemma}
\begin{proof}
Let $G'$ be $\mathbf{reduced}(R)$ and $V'$ be core vertices $V(G')\setminus S$.
By Lemma~\ref{lem:niceFVS} it holds $(X'\setminus R)\subseteq V'$.
We classify the vertices $V'$ as follows (see Fig.~\ref{fig:fvsCore_Reduced}):
\begin{itemize}
  \item $A = \left\{v\in V'\mid |\delta_{G'}(\{v\},S)| \geq 2 \right\}$.
  \item $B = \left\{v\in V'\setminus A\mid |N(v)|\cap(V'\setminus S)\geq
  3\right\}$.
  \item $C = V'\setminus(A\cup B)$.
\end{itemize}
We prove upper bounds for $A$, $B$ and $C$ separately.

To prove the upper bound for $A$, we consider the subgraph $G_A=(A\cup S,
\delta_{G'}(A,S))$ for $G'$.
If $A\cup S$ is a forest, $|E(G_A)| < |V(G_A)|$. Thus $2|A| \leq |E(G_A)| <
|V(G_A)| = |S| + |A|$ and therefore $|A| < |S|$.
Since we can make $G'$ a forest removing at most $|S|$ vertices disjoint from
$S$, $|A| < 2|S|$.

To prove the upper bound for $B$, we consider the forest $G[V']$.
Observe that all leaves of the forest are from $A$.
In fact if $v$ is a leaf of the forest, $d_{G[V']}(v)\leq 1$, and since
$d_{G'}(v)\geq 3$, $|\delta_{G'}(\{v\},S)| \geq 2$ must hold.
Each vertex in $B$ is an internal node of degree at least three in $G[V']$.
The number of such vertices cannot be more than the number of the leaves,
therefore $|B| \leq |A| < 2|S|$.

Each vertex $v$ in $C$ has a degree two in $G'[V']$ and exactly one incident
edge to $S$. Hence, $G[C]$ is composed of paths and isolated vertices.
Let $P$ be the path vertices and $I$ be the isolated vertices. We separately
bound the number of $P$ and $I$.

For each isolated vertex $v\in I$, the number of edges between $v$ and $A\cup B$
is exactly two. Since $G[V']$ is forest, $2|I| \leq |E(G[A\cup B\cup I])| <
|A\cup B\cup I| < 4|S| + |I|$, and therefore $|I|<4|S|$.
 
To prove the upper bound for $P$, we consider the graph $G[S\cup P]$. There are
exactly $|P|$ edges between $S$ and $P$ and at least $|P|/2$ edges among $G[P]$.
Thus the number of edges in $G[S\cup P]$ is at least $3|P|/2$.
If it is a forest, $3|P|/2 \leq |E(G[S\cup P])| < |V(G[S\cup P])| = |S| + |P|$,
hence $|P| < 2|S|$.
Since we can make $G'$ a forest removing at most $|S|$ vertices disjoint from
$S$ and removing a vertex evicts at most three vertices from $P$, we
obtain that $|P| < 5|S|$.

Altogether, $|V'| = |A|+|B|+|I|+|P|<|S|+2|S|+2|S|+4|S|+5|S|=13|S|$.
\end{proof}

Using Lemma~\ref{lem:niceFVS} and~\ref{lem:boundFVS}, we can solve
Task~\ref{task:disjFVS}.
First we construct the reduced graph in $O(|E|)$ time. Then if
the size of the graph is more than $14|S|$, we safely return `NO', and otherwise
we solve \textsf{Disjoint Feedback Vertex Set} for the reduced graph in
$O(f(k))$ time using some algorithm.
When $k$ is small or a fixed constant, the bottleneck is
the part of reducing the graph in $O(|E|)$ time. To speed up the part, we make
use of operations on link-cut trees.
\subsection{Faster Reduction} \label{sec:reduceFVS}
Let $G'$ denote $\mathbf{reduced}(R)$.
In this section, we show the algorithm computing $G'$ in
$O(k^3d^3\log n)$ amortized time by finding the core without explicitly reducing
the graph.
The graph $G[V\setminus X]$ is a forest. For vertices $u,v,w\in V\setminus
X$ that are in the same tree, let $\mathbf{meet}(u,v,w)$ denote the vertex at
which the path from $v$ to $u$ and the path from $w$ to $u$ firstly meet.
In other words, $\mathbf{meet}(u,v,w)$ is $\mathbf{nca}(v,w)$ on
the tree rooted at $u$, thus is computable in $O(\log n)$ amortized time.

To generate the core vertices $C$, we iterate over every set of three
edges $\{e_1,e_2,e_3\}\subseteq \delta(S, V\setminus X)$ and add
$\mathbf{meet}(v_1,v_2,v_3)$ to $C$, where $v_i$ is the endpoint of $e_i$ that
is in $V\setminus X$.
\begin{lemma}
The vertices $C$ generated by the above algorithm is equal to core
vertices, and the algorithm runs in $O(k^3d^3\log n)$ amortized time.
\end{lemma}
\begin{proof}
The time complexity is straightforward from the above discussion.
Let $v$ be a vertex in $V\setminus X$.
The vertex $v$ is a core vertex if and only if there are at least three
edge disjoint paths from $v$ to $S$ that contains no internal vertex in $X$.
If $v$ is a core vertex and there are such paths $P_1,P_2$ and $P_3$ to
$s_1,s_2,s_3\in S$, $v$ must be in $C$ since $v$ is the vertex at which the
path from $s_1$ to $s_3$ through $P_1$ and $P_3$ and the path from $s_2$ to
$s_3$ through $P_2$ and $P_3$ firstly meet.
Conversely, if $v$ is in $C$, it directly means there are at least three edge
disjoint paths from $v$ to $S$ that contains no internal vertex in $X$.
\end{proof}
After finding the core vertices, we can also compute the edges in $G'$ in
$O(k^3d^3\log n)$ amortized time using link-cut tree operations as follows.

Let us construct the graph $G''$ to be the reduced graph $G' =
\mathbf{reduced}(R)$.
The vertex set of $G''$ is $C\cup S$.

Firstly, let us consider the edges in $G'[C]$.
We iterate over every $u,v\in C$.
The edge $\{u,v\}\subseteq C$ is in $G'$ if and only if there is a path from $u$
to $v$ in the forest $G[V\setminus X]$ and it contains no internal vertex in $C$.
There is a path from $u$ to $v$ if and only if $u$ and $v$ are in the same tree.
To check if the path contains an internal vertex in $C$, 
firstly we make the vertex $u$ a root of the tree containing $u$ calling
$\mathbf{evert}(u)$.
Then we iterate over every $w\in C$ without $u$ and $v$.
For each $w\in C$, we cut the edge between $w$ and its parent and test if $u$
and $v$ are still in the same tree. If $u$ and $v$ are now separated, it means $w$
were on the path from $u$ to $v$. After checking, we restore the cut edge.
If there is no internal vertex $w$, we add the edge $\{u,v\}$ to $G''$. Since
$|C|\leq 14|S|$, using the algorithm, we complete the construction of $G''[C]$ in
$O(k^3\log n)$ amortized time.

Let us consider the edges between $S$ and $C$.
Observe that each edge from $s\in S$ to $v\in C$ in $G'$ corresponds
to a path from $s$ to $v$ that have no internal vertex in $C\cup S$.
Using the algorithm similar to the previous one, we can compute the edges in
$O(k^3d\log n)$ time.
Firstly we iterate over every $\{s,u\}\in \delta_G(S,V\setminus X)$ and $v\in
C$, where $s\in S$ and $u\in V\setminus X$.
We call $\mathbf{evert}(u)$, and for each $w\in C$ without $u$, we
check if $w$ is an internal vertex of the path from $s$ to $u$ in $O(\log n)$
amortized time.
If there is no internal vertex in $C$, we add the edge $\{s,v\}$ to $G''$.
Since $|\delta_G(S,V\setminus X)|=O(kd)$, using the algorithm, we complete
the construction of the edges between $S$ and $C$ in $O(k^3d\log n)$ amortized time.

Finally let us consider the edges in $S$. 
We add every edge in $G[S]$ to $G''$ in $O(k^2)$ time.
Furthermore, for each $s,t\in S$, if there is a path of length more than 1 from
$s$ to $t$ whose internal vertices not containing a vertex in $C\cup S$, it becomes an edge	
between $s$ and $t$.
To find the paths, we iterate over every set of two edges $\{e_1,e_2\}\subseteq
\delta(S, V\setminus X)$.
Let $e_i=\{s_i,v_i\}$, $s_i\in S$ and $v_i\in V\setminus X$.
If $v_1$ and $v_2$ are in the same tree, there is a
path $P$ from $v_1$ to $v_2$ in $G[V\setminus X]$. To check if $P$ contains
core vertex or not, we iterate over every edge $e_3$ in
$\delta(S,V\setminus X)$ without $e_1$ and $e_2$.  Let $v_3$ be the endpoint of
$e_3$ that is in $V\setminus X$.
If $v_3$ and $v_1$ is in the same tree, we can conclude $P$ contains core
vertex $\mathbf{meet}(v_1,v_2,v_3)$.
If there is no such edge, we can conclude $P$ contains no core vertex by the
construction of the core vertices, and add $\{v_1,v_2\}$ to $G''$.
It runs in $O(k^3d^3\log n)$ amortized time.

Now, we have completed the construction of the reduced graph $G'' = G'$.
Altogether, we can construct the reduced graph in $O(k^3d^3\log n)$
amortized time.  


Now we are ready to solve Task~\ref{task:disjFVS}.
First, we compute $G' = \mathbf{reduced}(R)$ in $O(k^3d^3\log n)$ amortized
time. If $|V(G')|\geq 14|S|$, we return `NO'.
Otherwise, for example, using the $O(3.83^kk|V|^2)$-time algorithm proposed by
Cao, Chen and Liu~\cite{DBLP:conf/swat/CaoCL10}, we solve the problem in
$O(3.83^kk^3)$ time.

Together with the outer loop cost $O(2^k)$, we finish the compression
phase in $O(7.66^kk^3 + 2^kk^3d^3\log n)$ amortized time.
\subsection{Exchange Phase}
Given a minimum feedback vertex set $X'$, the remaining work is to exchange the
solution $X$ to $X'$.
Let $R$ be $X\cap X'$, $S$ be $X\setminus R$ and $S'$ be $X'\setminus R$.
First, we add every vertex in $S'$ to $X$ as in Sectiondix~\ref{sec:fvsGrowing}.
Then, for each vertex $u$ in $S$ we remove $u$ from $X$ and add $u$ to the
link-cut trees calling $\mathbf{evert}(v)$ and $\mathbf{link}(u,v)$ for each
$v$ in $N(u)\setminus X$.
Since $|X|\leq k+1$ and $|X'|\leq k$, the amortized time complexity of this
phase is $O(kd\log n)$.

Altogether, we have obtained an $O(7.66^kk^3 + 2^kk^3d^3\log n)$-time dynamic
graph algorithm for \FVS, where $d$ is a degree bound on the graph.

\section{Acknowledgement}
Yoichi Iwata is supported by Grant-in-Aid for JSPS
Fellows (256487). Keigo Oka is supported by JST,
ERATO, Kawarabayashi Large Graph Project.

\bibliographystyle{abbrv}
\bibliography{paper}

\end{document}